\documentclass[conference]{IEEEtran}
\IEEEoverridecommandlockouts
\usepackage{tikz}
\usetikzlibrary{calc}
\usepackage{caption}
\usepackage{subcaption}
\usepackage{amsthm}
\usepackage{cite}
\usepackage{amsmath,amssymb,amsfonts}
\usepackage{mathtools}
\usepackage{algorithmic}
\usepackage{graphicx}
\usepackage{textcomp}
\usepackage{xcolor}
\usepackage{bbold}
\usepackage{mathrsfs}
   \newtheorem{lemma}{Lemma}
   \newtheorem{theorem}{Theorem}
   \usepackage{titlesec}
  \titlespacing\subsection{0pt}{0pt plus 1pt minus 2pt}{0pt plus 1pt minus 2pt}
   \titlespacing\section{0pt}{5pt plus 1pt minus 2pt}{5pt plus 1pt minus 2pt}
\begin{document}

\title{On the Performance of Large-Scale Wireless Networks in the Finite Block-Length Regime}
\author{
\IEEEauthorblockN{Nourhan Hesham and Anas Chaaban}
\IEEEauthorblockA{School of Engineering, University of British Columbia, Kelowna, BC V1V1V7, Canada\\
Email: \{nourhan.soliman,anas.chaaban\}@ubc.ca
}
\thanks{%
This publication is based upon work supported by King Abdullah University of Science and Technology (KAUST) under Award No. OSR-2018-CRG7-3734.
}
}

\maketitle
\begin{abstract}
 Ultra-Reliable Low-Latency Communications have stringent delay constraints, and hence use codes with small block length (short codewords). In these cases, classical models that provide good approximations to systems with infinitely long codewords become imprecise. To remedy this, in this paper, an average coding rate expression is derived for a large-scale network with short codewords using stochastic geometry and the theory of coding in the finite blocklength regime. The average coding rate and upper and lower bounds on the outage probability of the large-scale network are derived, and a tight approximation of the outage probability is presented. Then, simulations are presented to study the effect of network parameters on the average coding rate and the outage probability of the network, which demonstrate that results in the literature derived for the infinite blocklength regime overestimate the network performance, whereas the results in this paper provide a more realistic performance evaluation.  
\end{abstract}

\begin{IEEEkeywords}
Stochastic Geometry; Large-Scale Network; Capacity; Outage Probability; Finite Blocklength; URLLC.
\end{IEEEkeywords}

\section{Introduction}
The density of cellular networks has increased significantly from 2G up to 5G, and continues to increase in order to serve a larger number of users/devices, and provide wider coverage and higher data speeds. Additionally, current and future networks are expected to support a multitude of connectivity requirements, including the Internet-of-Things (IoT) and Machine-Type communications (MTC). Many such applications have stringent delay requirements, which necessitate using different approaches in studying performance \cite{BennisDebbahPoor2018}. Recent works on Ultra-Reliable Low-Latency Communications (URLLC) focused on this topic, investigating low-latency communications from different perspectives \cite{EEE,HouSheLiVucetic2019,
HindiElayoubiChahed,Park2019,ChaabanSezgin2016}. 

Large-scale networks can be studied using stochastic geometry (SG) tools \cite{sawy}. SG is the study of random spatial patterns. It is a strong tool used for interference modeling in large-scale networks, and has been used to study several network performance metrics in the literature~\cite{paper2_SG,paper3_SG,A,B}. In~\cite{paper2_SG}, the wireless network is modelled using SG as the Poisson point process (PPP). This led to results on the connectivity, the capacity, the outage probability, and other fundamental limits of wireless networks. In~\cite{paper3_SG}, the coverage probability of cellular networks in urban areas modelled as a PPP is provided. The results are compared to the hexagonal grid model to find that the SG model provides a more accurate upper bound on the coverage probabilities than when using the hexagonal grid model. In~\cite{A,B}, large-scale networks using non-orthogonal multiple-access have been analyzed using SG. Note that works in this area commonly use Shannon's channel capacity expression \cite{Shannon} to study performance which is not suitable for delay-limited applications. 

Aggregate interference modeling and performance characterization were active research topics for decades. Due to the difficulty in modeling aggregate interference, many papers provided approximations to be able to reach a tight approximation to the capacity, symbol error probability, outage probability, etc. A common method is to approximate the aggregate interference as a Gaussian random variable as in~\cite{EID,sawy,gaussian3}. In~\cite{EID}, the aggregate interference was approximated as a sum of Gaussian random variables with random scaling. In~\cite{sawy}, the authors modified the work in~\cite{EID} to approximate the aggregate interference as a single Gaussian random variable with random scaling giving the same results as in~\cite{EID}. Moreover in~\cite{gaussian3}, the authors provide the kurtosis of the interference distribution for different values of exclusion regions which shows that the interference tends to be Gaussian for large exclusion regions. As a conclusion from different papers, the Gaussian approximation is a valid approximation for dense networks or large exclusion regions. 

Since that delay-constrained applications require the use of short codes. Studying the achievable information rate in such applications using Shannon's capacity expression becomes imprecise, as this expression is derived for infinitely long codes and vanishing error probability. Instead, the coding rate under a codelength limitations has to be used for such studies.
In~\cite{polyanski}, Polyanskiy {\it et al.} proposed tight bounds on the maximal channel coding rate achievable for a given blocklength regime (short codewords) for different types of channels  such as the binary symmetric channel (BSC), binary errasure channel (BEC), and additive white gaussian noise (AWGN) channel. Moreover, in~\cite{BlockFadingPolyankiy}, Polyankiy {\it et al.} extended~\cite{polyanski} to include the maximal achievable coding rate over block-fading channels in a finite blocklength regime. These works were extended to studying the coding rate in the finite blocklength regime of other scenarios, including relaying~\cite{Relaying}, MTC~\cite{MTC}, and multiaccess communication in~\cite{MultiAccess}. However, to the best of our knowledge, there are no works apart from \cite{Park2019} in the literature which derives the decoding error for large-scale networks in the finite blocklength regime. In this paper, we derive the average coding rate of large-scale networks in the finite blocklength regime. We also formulate the outage probability of the network in the finite blocklength regime, and derive bounds and a fairly tight approximation on this outage probability. Then, we investigate the effect of network parameters on performance, and demonstrate clearly how the Shannon's capacity expression for the infinite blocklength regime overestimates performance. These results are applicable for studying IoT networks, MTC, etc.

The rest of the paper is organized as follows. In Sec. II, the system model is presented. In Sec. III, the average capacity of a large-scale network in the finite blocklength regime is derived, and the outage probability is formulated. Finally, the effect of network parameters on the capacity is investigated in Sec. IV, and the paper is concluded in Sec. V.

\section{System Model}
We consider a downlink scenario where a serving base station (BS) transmits to a user located within its coverage area, interfered by other non-serving BSs as shown in Fig.~\ref{fig:r_o_r_i}. The channel is modeled as a block-fading Rayleigh channel. The received signal is given by
\begin{equation}\label{eq:system_model}
   y=h_0 \sqrt{P} r_0^{-\eta/2} s_0+I_{gg}+w, 
\end{equation}
where the channel gain $(h_0)$ is circularly symmetric complex Gaussian ($\mathcal{CN}(0,1)$), $P$ is the transmit power, $r_0$ is the distance between the BS and the user, $\eta$ is the path loss exponent, $s_0$ is a codeword symbol with unit power, $I_{gg}$ is the interference from other BSs, and $w$ is $\mathcal{CN}(0,N_0)$. The channel is considered to be a block fading channel model where the channel coefficient ($h_0$) remains constant for a block of L consecutive symbols and changes to an independent realization in the next block. The interference term $I_{gg}$ is the sum of interference signals received from all non-serving BSs and is given by
\begin{equation}
    I_{gg}=\sum_{i=1}^{\infty}\sqrt{P} h_i r_i^{-\eta/2} s_i,
\end{equation}
where $P$ is the transmit power (assumed equal across BSs), $h_i$ is the block-fading channel gain from non-serving $\mathrm{BS}_i$ to the user, $r_i$ is the distance between non-serving $\mathrm{BS}_i$ and the user, and $s_i$ is the codeword symbol transmitted by non-serving $\mathrm{BS}_i$.

 It is worth noting that $r_0$ is random in general, but for simplicity it is assumed to be fixed in this work to study the performance for different system parameters. We assume that there is no interfering (non-serving) BS within a circle of radius $r_0$ about the user, which is known as the interference exclusion region~\cite{sawy}. To study the average performance of such a network over different geometries, the BS locations are often modeled by a repulsive point process (PP)~\cite{sawy}. For tractability, a PPP is commonly considered as an accurate approximation for several types of intractable repulsive PP, and the reader is referred to~\cite{sawy} for more details. Hence, in this work, we assume that the BSs locations follow a PPP with an interference exclusion region of radius $r_0$ and intensity $\lambda$ $\mathrm{BS/km^2}$.

\begin{figure}
\centering
\begin{tikzpicture}[scale = 0.45]
\node (u) at (0,0) {\includegraphics[scale=.05]{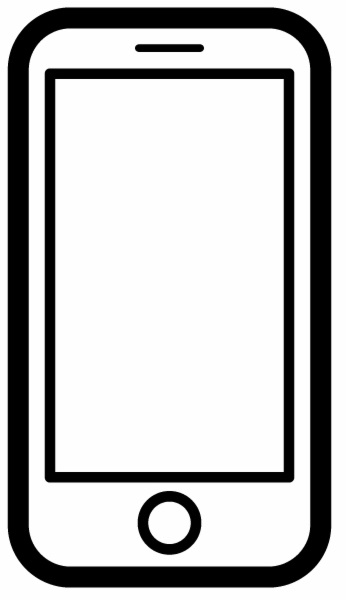}};
\node (bs1) at ($(u)+(7,0)$) {\includegraphics[scale=.016]{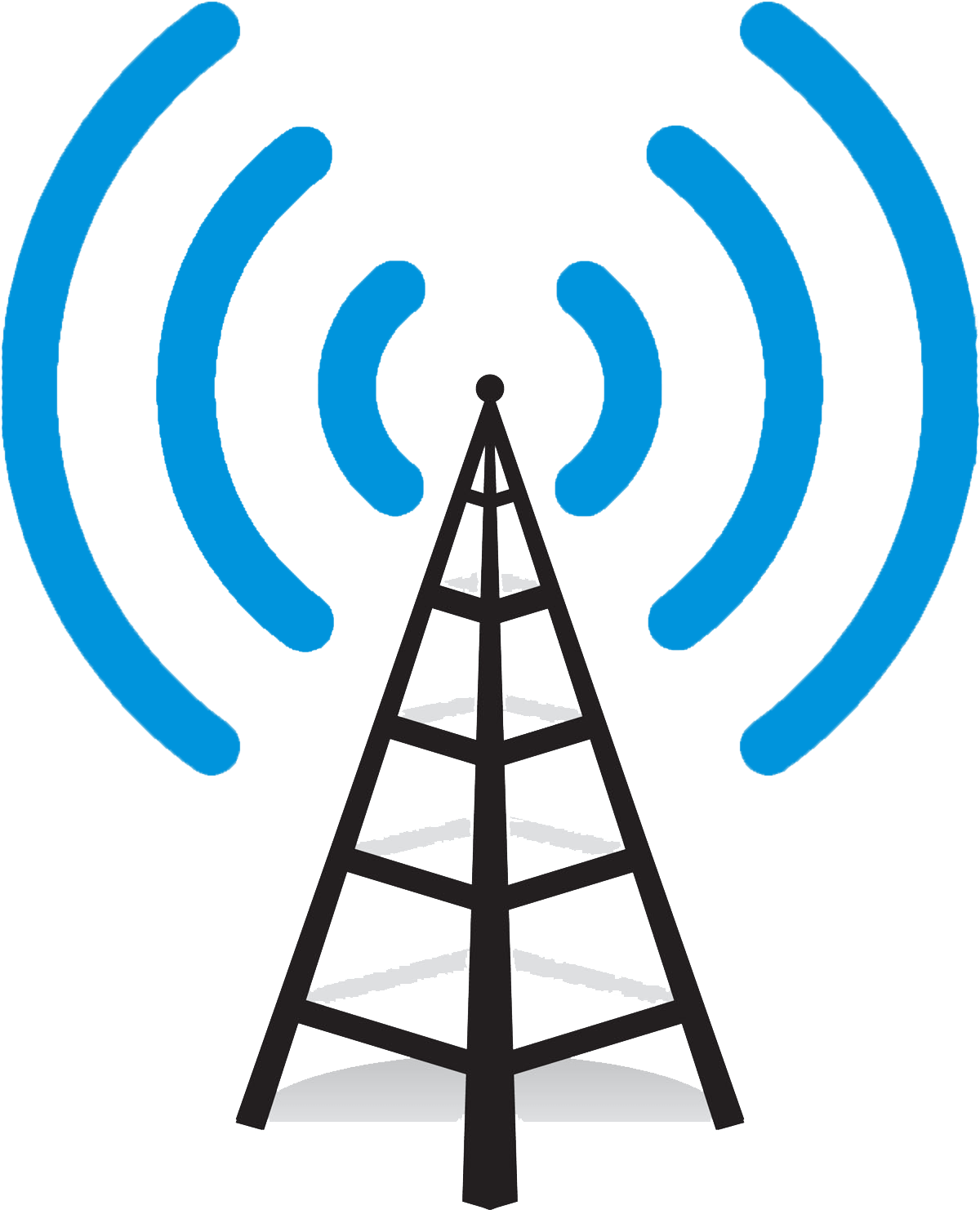}};
\node (bs7) at ($(u)+(-6,3)$) {\includegraphics[scale=.014]{bs.png}};
\node (bs0) at ($(u)+(3.5,1)$) {\includegraphics[scale=.015]{bs.png}};
\node (bs4) at ($(u)+(-3,-4)$) {\includegraphics[scale=.019]{bs.png}};
\node (bs5) at ($(u)+(-7,-4.5)$) {\includegraphics[scale=.02]{bs.png}};
\node (bs8) at ($(u)+(-2,4)$) {\includegraphics[scale=.013]{bs.png}};
\node (bs2) at ($(u)+(5,-4)$) {\includegraphics[scale=.019]{bs.png}};
\node (bs3) at ($(u)+(1,-4.5)$) {\includegraphics[scale=.02]{bs.png}};
\node (bs6) at ($(u)+(-5,-1)$) {\includegraphics[scale=.017]{bs.png}};
\node (bs9) at ($(u)+(3,4)$) {\includegraphics[scale=.013]{bs.png}};
\draw[->,line width=1] (bs0) to node[above] {$r_0$} (u);
\draw[->,line width=1,dashed] (bs1) to node[below] {$r_1$} (u);
\draw[->,line width=1,dashed] (bs2) to node[above] {$r_2$} (u);
\draw[->,line width=1,dashed] (bs3) to node[right] {$r_3$} (u);
\draw[->,line width=1,dashed] (bs4) to node[right] {$r_4$} (u);
\draw[->,line width=1,dashed] (bs5) to node[above] {$r_5$} (u);
\draw[->,line width=1,dashed] (bs6) to node[above] {$r_6$} (u);
\draw[->,line width=1,dashed] (bs7) to node[right] {$r_7$} (u);
\draw[->,line width=1,dashed] (bs8) to node[right] {$r_8$} (u);
\draw[->,line width=1,dashed] (bs9) to node[left] {$r_9$} (u);
\node[circle,red,dashed] (c) at (u) [draw,minimum width = 3cm] {};
\node (t) at ($(u)+(6,4)$) [text width=1.5cm,red] {\footnotesize Interference exclusion region};
\draw[->,red] (t) to (c);
\end{tikzpicture}	
	
	\caption{A realization of a cellular network with exclusion region of $r_0$ between the user of interest and the serving BS, where BSs are randomly distributed with distances $r_i> r_0$ to the user.}
	\label{fig:r_o_r_i}
\end{figure}

Studying the network performance (Capacity, BER, etc.) under this model is rather difficult. To remedy this, an approximation is commonly considered in the literature, wherein the interference is modeled as conditionally Gaussian, conditioned on geometry ($r_0$ and $r_i$)~\cite{sawy}. Thus, the simplified interference representation is a randomly scaled Gaussian given by:
\begin{align}\label{eq:I_eq}
 I_{eq}=\sqrt{\mathcal{B}} {G},
 \end{align}
where ${G}$ is $\mathcal{CN}(0,1)$, and $\mathcal{B}>0$ is the power of interference, and it is a random variable independent of ${G}$ but is dependent on the network geometry, and has the following Laplace transform (LT)
\begin{align}
\label{LapTrans}
    \mathcal{L}_\mathcal{B}(z)=\exp\left\{ \sum_{k=1}^{\infty} a_k z^k\right\},
\end{align}
where the coefficients $a_k$ are given by
\begin{align}
       a_k=(-1)^k 2 \pi \lambda r_0^2 \left(\frac{P}{r_0^\eta}\right)^k \frac{\mathbb{E}\big\{|s_0|^{2k}\big\}}{(\eta k -2)k!}.
\end{align}
For $\eta=4$, the LT expression provided in~\eqref{LapTrans} can be simplified to
\begin{align}\label{eq:laplace_z}
\mathcal{L}_{\mathcal{B}}(z)=\exp\left(-\pi \lambda \sqrt{z P} \arctan\left(\frac{\sqrt{z P}}{r_0^2}\right)\right).
\end{align}
 Using the approximation in~\eqref{eq:I_eq}, the simplified model becomes 
\begin{align}
y=h_0 \sqrt{P} r_0^{-\eta/2} s_0+I_{eq}+w. 
\end{align}

We assume that the serving BS sends information to the user using codewords of length $n$ symbols, and that the codewords have to be decoded correctly with probability $(1-\epsilon)$ where $\epsilon$ is the frame error probability. The goal of this work is to study the average coding rate (in bits per transmission) and outage probability of this network under these considerations. The average coding rate is discussed in the next section.


\section{  Average Coding Rate in the finite blocklength regime} \label{Capacity_proof}
\setlength{\abovedisplayskip}{7pt}
\setlength{\belowdisplayskip}{7pt}
\setlength{\abovedisplayshortskip}{7pt}
\setlength{\belowdisplayshortskip}{7pt}
In this section, the average coding rate in the finite blocklength regime is derived. We denote the average coding rate of the network given blocklength $n$ and frame error rate $\epsilon$ by $\bar{R}_{n,\epsilon}$. Assuming channel knowledge is available at the BSs, the signal to interference and noise ratio $(\gamma)$ is defined as
\begin{align}\label{SINR}
    \gamma&=\frac{P r_0^{-\eta} |h_0|^2}{N_0+\mathcal{B}}= \frac{ |h_0|^2}{\frac{N_0}{P r_0^{-\eta}}+\frac{\mathcal{B}}{P r_0^{-\eta}}}\nonumber\\
    &= \frac{ |h_0|^2}{\gamma_o+\zeta}
\end{align}
where $\gamma_o=\frac{P r_0^{-\eta}}{N_0}$ and  $\zeta=\frac{\mathcal{B}}{P r_0^{-\eta}}$. The average capacity of the large-scale network in the infinite blocklength regime is given by~\cite{sawy} 
\begin{align}
   \bar{C}_{\infty,0}(\gamma_o)&=\mathbb{E}_{h_0,\zeta},\Big\{\log_2(1+\frac{ |h_0|^2}{\gamma_o+\zeta})\Big\},\nonumber
\end{align}

To derive the average coding rate in the finite blocklength regime, i.e., where $n<\infty$ and $\epsilon>0$, we rely on a result in~\cite{BlockFadingPolyankiy} concerning the maximum coding rate of a block fading channel in the finite blocklength regime, which is introduced in the following lemma.

\begin{lemma}{\textsc{(\cite{polyanski})}}\label{lem:poly}
For a block fading channel with a signal-to-noise ratio $\alpha=\frac{P}{N_0}$, blocklength $n$, a target frame error rate $\epsilon$ satisfying $0<\epsilon<0.5$, and channel gain coefficient $H$ , the maximum coding rate is approximated as
\begin{equation}\label{eq:Cap_polyankiy}
    R_{n,\epsilon}(H,\alpha)\approx C_{\infty,0}(H,\alpha)-\frac{\sqrt{V(H,\alpha)} Q^{-1}(\epsilon)}{\sqrt{n}}+o(1/\sqrt{n}),
\end{equation}
where $C_{\infty,0}(H,\alpha) = \log_2(1+|H|^2\alpha)$ is the AWGN capacity in the infinite blocklength regime, $V(H,\alpha)$ is the channel dispersion given by 
\begin{align}
      V(H,\alpha)=\frac{|H|^2\alpha}{2} \frac{|H|^2\alpha+2}{(|H|^2\alpha+1)^2} \log_2^2(e),
\end{align}
and $Q(\cdot)$ is the Q-function.
\end{lemma}

In what follows, we treat this approximation as the true maximum coding rate, since this approximation is accurate enough for practical values of $n$ as demonstrated in \cite{BlockFadingPolyankiy}. To extend this to the average coding rate expression of the large-scale network in the finite blocklength regime, $|H|^2\alpha$ should be replaced by $\gamma$ defined in~(\ref{SINR}), followed by averaging with respect to $\gamma$, to obtain the average capacity of the large-scale network in the infinite blocklength regime $(\bar{C}_{\infty,0}(\gamma_o))$ and the average channel dispersion $(\bar{V}(\gamma_o))$  which are discussed in the next subsections. This extension is valid since the Gaussian approximation is considered for the aggregate interference.
\vspace{5pt}
\subsection{Average Capacity in the Infinite blocklength Regime for large-scale network}
\vspace{2pt}
The average capacity of the large-scale network in the infinite blocklength regime with channel state information available at the BSs is given by the following lemma~\cite{sawy}.
\begin{lemma}\label{lem:capacity}
For a large-scale network topology with an average power constraint block-fading Rayleigh channel, signal to noise ratio $\gamma_o$ and interference power $\zeta$, the average capacity is given by:
\begin{align}\label{eq:avg_C_network}
\bar{C}_{\infty,0}(\gamma_o)=\int_{0}^{\infty} \exp\Big(-\frac{2^{ c}-1}{\gamma_o}\Big) \mathcal{L}_\zeta\Big\{(2^{c}-1)\Big\} d c
\end{align}
\end{lemma}
This result was derived in~\cite{sawy} to which the reader is referred for the proof. Next, we derive the average channel dispersion.

\subsection{Average Channel Dispersion for a Large-Scale Network}
The average channel dispersion for the large-scale network is given by
\begin{align}
  \bar{V}(\gamma_o)= \mathbb{E}\left\{ V(\gamma)\right\}&= \mathbb{E}\left\{\frac{\gamma}{2} \frac{\gamma+2}{(\gamma+1)^2} \log_2^2(e)\right\}\nonumber\\
    &= \mathbb{E}\left\{ \frac{\log_2^2(e)}{2} \left(1-\frac{1}{(\gamma+1)^2}\right)\right\}\nonumber\\
    &=\int_{0}^{\infty} (1-\mathbb{F}_V({v})) d{v}\label{eq:cdf_v}
\end{align}
where $\mathbb{F}_V(v)$ is the cumulative density function (CDF) of the channel dispersion $V(\gamma)$, and the last step follows as an application of Fubini's theorem~\cite{fubini}. The following lemma expresses $\bar{V}(\gamma_o)$.
\begin{lemma}\label{Lemma_EV}
The average channel dispersion $\mathbb{E}\{V(\gamma)\}$ is given by 
\begin{align}\label{eq:avg_V_network}
 \bar{V}(\gamma_o)=\int_{0}^{\frac{1}{2}\log_2^2(e)} \exp\left(-\frac{1}{\gamma_o}\left(\sqrt{\frac{1}{1-\frac{2{v}}{\log_2^2(e)}}}-1\right)\right)\nonumber\\
    \times\mathcal{L}_\zeta\left\{\sqrt{\frac{1}{1-\frac{2{v}}{\log_2^2(e)}}}-1\right\} d{v},
\end{align}
where $\mathcal{L}_{\zeta}\{\cdot\}$ is given in \eqref{LapTrans} and $\zeta=\frac{\mathcal{B}}{P r_0^{-\eta}}$.
\end{lemma}
\begin{proof}
We start by expressing the CDF of $V(\gamma)$ for a given $\zeta$ as follows:
\begin{align}
   & \mathbb{F}_V(  {v}|\zeta)= \mathbb{P}(V(\gamma)<  {v}|\zeta)\nonumber\\
    &=\mathbb{P}\left(\frac{\log_2^2(e)}{2} \left(1-\frac{1}{(\gamma+1)^2}\right)<   {v}\right)\nonumber\\ 
     &=\mathbb{P}\left(\gamma<\sqrt{\frac{1}{1-\frac{2  {v}}{\log_2^2(e)}}}-1\right)\nonumber\\
     &=\mathbb{P}\left(\frac{|h_0|^2}{\frac{1}{\gamma_o}+\zeta}<\sqrt{\frac{1}{1-\frac{2  {v}}{\log_2^2(e)}}}-1\right)\nonumber\\
     &=\mathbb{P}\left(|h_0|^2<\left(\frac{1}{\gamma_o}+\zeta\right)\left(\sqrt{\frac{1}{1-\frac{2  {v}}{\log_2^2(e)}}}-1\right)\right)\nonumber\\
     &=1-\exp\left(-\left(\frac{1}{\gamma_o}+\zeta\right)\left(\sqrt{\frac{1}{1-\frac{2  {v}}{\log_2^2(e)}}}-1\right)\right),
\end{align}
where the last step follows from the Rayleigh distribution of $h_0$. Averaging with respect to the interference term $\zeta$ yields
\begin{align}
\label{eq:avg_cdf_v}
&\mathbb{E}_\zeta\{\mathbb{F}_V(  {v})\}\nonumber\\
&= \mathbb{E}_\zeta\left\{1-\exp\left(-\left(\frac{1}{\gamma_o}+\zeta\right)\left(\sqrt{\frac{1}{1-\frac{2  {v}}{\log_2^2(e)}}}-1\right)\right)\right\}\nonumber\\
& =1-\exp\left(-\frac{1}{\gamma_o}\left(\sqrt{\frac{1}{1-\frac{2  {v}}{\log_2^2(e)}}}-1\right)\right)\nonumber\\
&\hspace{3.5cm} \times\mathcal{L}_\zeta\left\{\sqrt{\frac{1}{1-\frac{2  {v}}{\log_2^2(e)}}}-1\right\}.
\end{align}

By substituting (\ref{eq:avg_cdf_v}) in (\ref{eq:cdf_v}), we obtain
\begin{multline}
   \bar{V}(\gamma_o)=\int_{0}^{\infty} (1-\mathbb{F}_V(  {v})) d  {v}\\
    =\int_{0}^{\frac{1}{2}\log_2^2(e)} \exp\left(-\frac{1}{\gamma_o}\left(\sqrt{\frac{1}{1-\frac{2  {v}}{\log_2^2(e)}}}-1\right)\right)\\
    \times\mathcal{L}_\zeta\left\{\sqrt{\frac{1}{1-\frac{2  {v}}{\log_2^2(e)}}}-1\right\} d  {v}.
\end{multline}
This proves the statement of the lemma.
\end{proof}
Using Lemma~\ref{lem:poly},\ref{lem:capacity},\ref{Lemma_EV}, we obtain the following theorem which expresses the average coding rate of the large-scale network in the finite blocklength regime.
\begin{theorem}
The average coding rate of the large-scale network modeled by \eqref{eq:system_model} with blocklength $n$, target frame error rate $\epsilon$, and signal-to-noise ratio $\gamma_0=\frac{r_0^{-\eta}P}{N_0}$ is given by
\begin{align}\label{eq:Capacity_FB}
    \bar{R}_{n,\epsilon}(\gamma_o)=\bar{C}_{\infty,0}(\gamma_o)-\frac{\sqrt{\bar{V}(\gamma_o)} Q^{-1}(\epsilon)}{\sqrt{n}}+o(1/\sqrt{n}),
\end{align}
where $\bar{C}_{\infty,0}(\gamma_o)$ and $\bar{V}(\gamma_o)$ are as defined in~(\ref{eq:avg_C_network})~and~(\ref{eq:avg_V_network}), respectively.
\end{theorem}
\begin{IEEEproof}
The result follows by averaging \eqref{eq:Cap_polyankiy} with respect to $h_0$ and $\mathcal{B}$, and using \eqref{eq:avg_C_network} and \eqref{eq:avg_V_network}.
\end{IEEEproof}

To achieve this rate, the transmitter uses a code of length $n$, and adapts the rate for each transmission block depending on the channel state $h_0$, which is assumed to be known at the transmitter. Note that this reproduces the result in the infinite blocklength regime when $n\to\infty$ since terms vanish as $n\rightarrow \infty$. Next, we discuss the outage probability in the finite blocklength regime.

\section{Outage Probability}
 Outage is defined as the event where the channel capacity is lower than a rate threshold corresponding to the target coding rate. In the infinite blocklength regime, this rate threshold can be converted to an $\mathrm{SINR}$ threshold. The outage probability of a large-scale network in the infinite blocklength regime is given in~\cite{sawy} as
\begin{align}\label{eq:outage}
 \mathcal{O}(r_0,T)=1-\exp\left(-\frac{T N_0 r_0^4}{P}\right) \mathcal{L}_\zeta(T) 
\end{align}
where $T$ is the $\mathrm{SINR}$ threshold on $\gamma$, i.e., an outage occurs when $\gamma$ is less than $T$. The reader is referred to~\cite{sawy} for complete proof.

In the infinite blocklength regime, when outage occurs, the channel is not guaranteed to support reliable communication at the target rate, where reliability is defined in the sense of a vanishingly small error probability. To extend this to the finite blocklength regime, we define outage in the finite blocklength regime as follows. We say that the channel is in outage when it is not guaranteed to support transmission at the target rate at the desired frame error rate $\epsilon$ and blocklength $n$. This occurs when the average coding rate in a finite blocklength regime drops below the rate threshold. 

For a large-scale network in the finite blocklength regime, the outage probability can be calculated using~\eqref{eq:Capacity_FB} as follows. Let $R=\log_2 (1+T)$ be the target rate. Then the outage probability is given by 
\begin{align}
   &\mathcal{O}(r_0,T,n,\epsilon)\nonumber\\
   &=\mathbb{P}(C_{n,\epsilon}(\gamma)<R)\nonumber\\
   &=\mathbb{P}\left(C_{\infty,0}(\gamma_o)-\sqrt{\frac{ V(\gamma) }{n}} Q^{-1}(\epsilon)+o\left(\frac{1}{\sqrt{n}}\right)<R\right).
   \end{align}
Let 
\begin{align}
\label{ab}
    a&= \sqrt{ \frac{\log_2^2(e)}{2 n} } Q^{-1}(\epsilon)  \text{ \ and \ }
    b=o\left(\frac{1}{\sqrt{n}}\right).
\end{align}
Then, we can write
\begin{align}
&\mathcal{O}(r_0,T,n,\epsilon)\nonumber\\
&=\mathbb{P}\left(\log_2 (1+\gamma)-a\sqrt{\left(1-\frac{1}{(1+\gamma)^2}\right)}+b<R\right).
\end{align}
Noting that $\sqrt{\left(1-\frac{1}{(1+\gamma)^2}\right)}$ is in $[0,1]$, we conclude that
\begin{align}
\mathcal{O}_{l}(r_0,T,n,\epsilon)\leq \mathcal{O}(r_0,T,n,\epsilon) \leq \mathcal{O}_{u}(r_0,T,n,\epsilon),
\end{align}
where the lower bound is given by
\begin{align}\label{eq:Outage_LB}
\mathcal{O}_{l}(r_0,T,n,\epsilon) &= \mathbb{P}\left(\log_2 (1+\gamma)+b<R\right)\nonumber\\
&=1-\exp\left(-\frac{(2^{R-b}-1)}{\gamma_o}\right) 
\mathcal{L}_\zeta(2^{R-b}-1),
\end{align}
and the upper bound is given by
\begin{align}
\mathcal{O}_{u}(r_0,T,n,\epsilon)
&=\mathbb{P}\left(\log_2 (1+\gamma)-a+b<R\right)\nonumber\\
&= 1\hspace{-.05cm}-\exp\left(\hspace{-.1cm}-\frac{(2^{R+a-b}-1) }{\gamma_o}\right)\mathcal{L}_\zeta(2^{R+a-b}\hspace{-.1cm}-1).
\end{align}
Neglecting $b=o\left(\frac{1}{\sqrt{n}}\right)$ in~\eqref{eq:Outage_LB} which is small for practical values of $n$ ($n>100$ e.g.), we can see that the outage probability lower bound coincides with the outage probability in the infinite blocklength regime~\eqref{eq:outage}, confirming that the outage probability in the finite blocklength regime is larger than that in the infinite blocklength regime. On the other hand, by observing the term $\sqrt{\left(1-\frac{1}{(1+\gamma)^2}\right)}$, we can see that this term quickly approaches one as $\gamma$ grows. This indicates that the upper bound $\mathcal{O}_{u}(r_0,T,n,\epsilon)$ provides a good approximation for the outage probability when $\gamma$ is reasonably large, as stated next.

\begin{theorem}
The outage probability of the large-scale network modeled by \eqref{eq:system_model} with blocklength $n$, target frame error rate $\epsilon$, and signal-to-noise ratio $\gamma_0=\frac{r_0^{-\eta}P}{N_0}$ can be approximated as 
\begin{align}
\label{eq:outage_app}
\mathcal{O}(r_0,T,n,\epsilon)\approx \mathcal{O}_{u}(r_0,T,n,\epsilon).
\end{align}
\end{theorem}
\begin{IEEEproof}
This follows by approximating $\sqrt{\left(1-\frac{1}{(1+\gamma)^2}\right)}$ by $1$ which is a good approximation for moderate/large values of $\gamma$.
\end{IEEEproof}

Next, the average coding rate and the outage probability expressions are evaluated for different system parameters.

\section{Numerical Results}
Four main parameters affect the average coding rate of the large-scale network: The distance between the user and the serving BS ($r_0$) which also determines the interference exclusion region, the density of BSs per square kilometer $\lambda$, the blocklength $n$, and the target frame error probability $\epsilon$. The effect of these parameters on the average coding rate and the outage probability is discussed next.

The blocklength $n$ depends on the application scenario. Recent IoT applications use short packets in the range of 512 bits up to 4096 bits according to~\cite{IOT}. Moreover, the URLLC systems require low frame error probabilities ($\epsilon$). The average coding rate of the network is investigated for a range of $n$ and $\epsilon$ in Fig.~\ref{fig:n}, where $n$ takes values from the set $\{128,\ 2048\}$ and $\epsilon$ takes values from the set $\{ 10^{-2},\ 10^{-8}\}$. As shown in Fig.~\ref{fig:n}, at a specific target frame error probability, the average coding rate increases as the $n$ increases moving towards the average capacity in the infinite blocklength regime. However, the average coding rate decreases as the target frame error probability decreases. Moreover, the impact of $n$ becomes more significant if the target frame error probability is small, which is the requirement of URLLC which makes the expression more practical. 

Fig.~\ref{fig:Outage_App} shows the outage probability verses $r_0$ at (a) $n=128$ and (b) $n=2048$ and $\lambda=1$ and $9\ \mathrm{BS/km^2}$, with $\epsilon=10^{-2}$ and $10^{-8}$. The figure shows the Monte Carlo simulated outage probability, the approximation in \eqref{eq:outage_app}, and the outage probability in the infinite blocklength regime \eqref{eq:outage}. It shows that the approximate outage probability expression in \eqref{eq:outage_app} is fairly accurate, providing a convenient approximation for broad range of $n,\lambda,\epsilon$. The figure shows that the outage probability increases as the target frame error probability decreases because the term $\sqrt{\frac{V}{n}}Q^{-1}(\epsilon)$ increases as the $\epsilon$ decreases. Hence, the average coding rate decreases and outage probability increases. However, at a specific $\epsilon$, as the blocklength increases, the outage probability decreases towards the outage of the infinite blocklength regime.

Both figures \ref{fig:n} and \ref{fig:Outage_App} show how the results based on the capacity expression for the infinite blocklength regime overestimate performance, whereas the average coding rate and outage probability results in this paper provide an accurate estimate.  

\begin{figure}
    \centering
    \includegraphics[width=\linewidth]{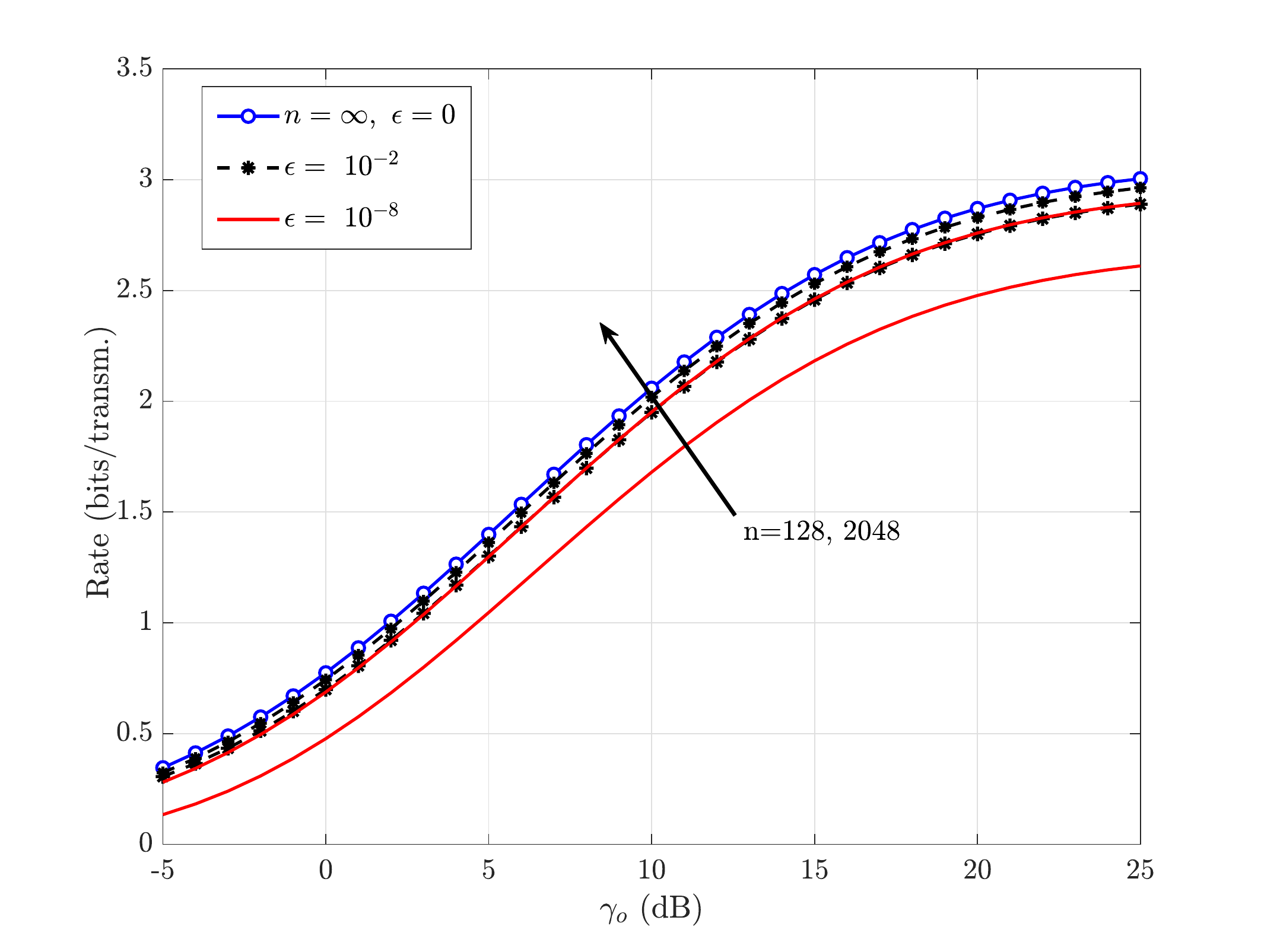}
    \caption{The average coding rate of the large-scale network vs. $\gamma_o$ for blocklengths $n$ where $n=\ 128,\ 2048$ at frame error probabilities $\epsilon=10^{-2},\ 10^{-8}$ with $\lambda=1\ \mathrm{BS/Km^2}$ and $r_0=250\ m$.}
    \label{fig:n}
\end{figure}

\begin{figure}
     \centering
     \begin{subfigure}[b]{0.5\textwidth}
         \centering
         \includegraphics[width=\linewidth]{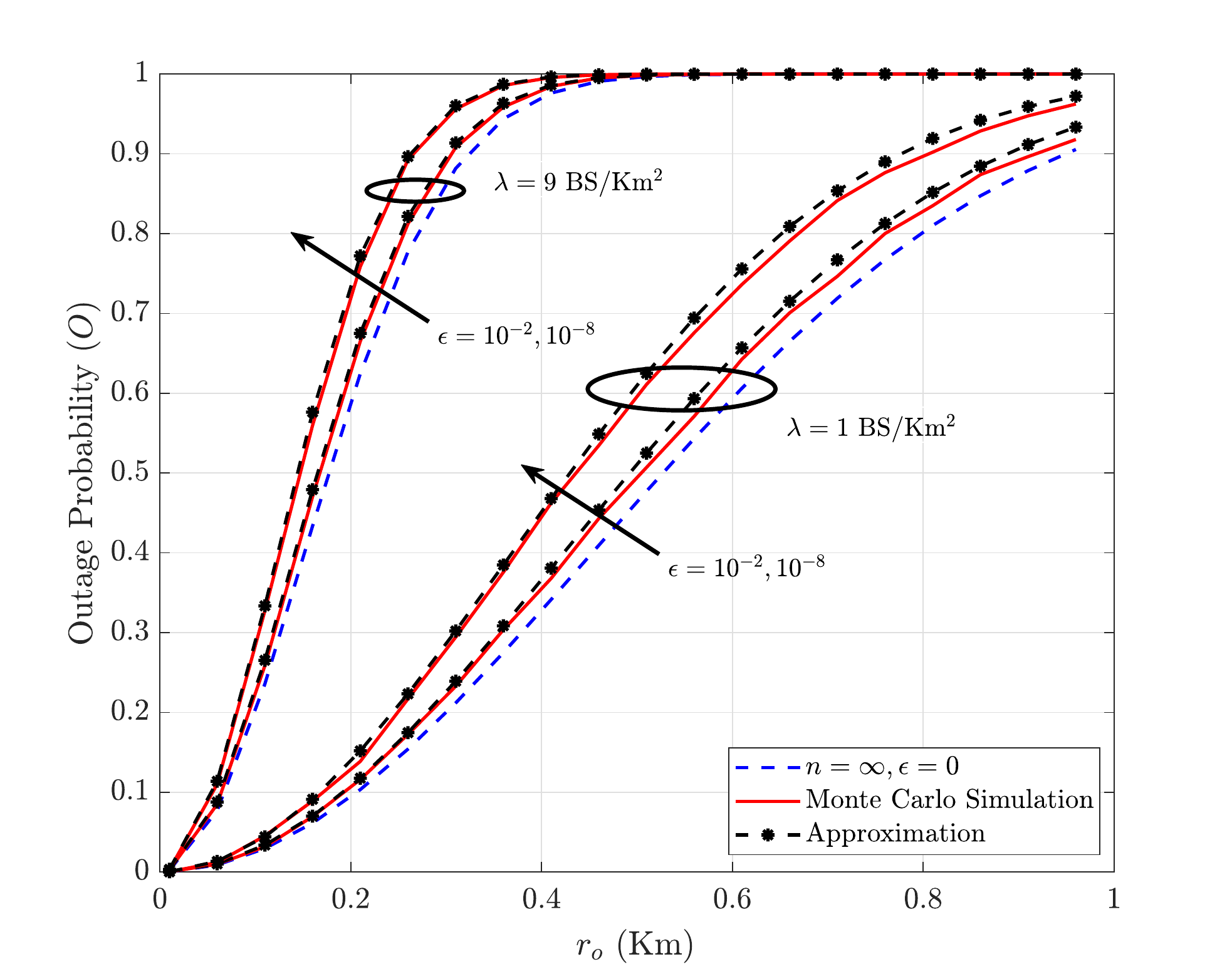}
         \caption{$n=128$}
         \label{fig:Outage_App_128}
     \end{subfigure}
     \hfill
     \begin{subfigure}[b]{0.5\textwidth}
         \centering
         \includegraphics[width=\linewidth]{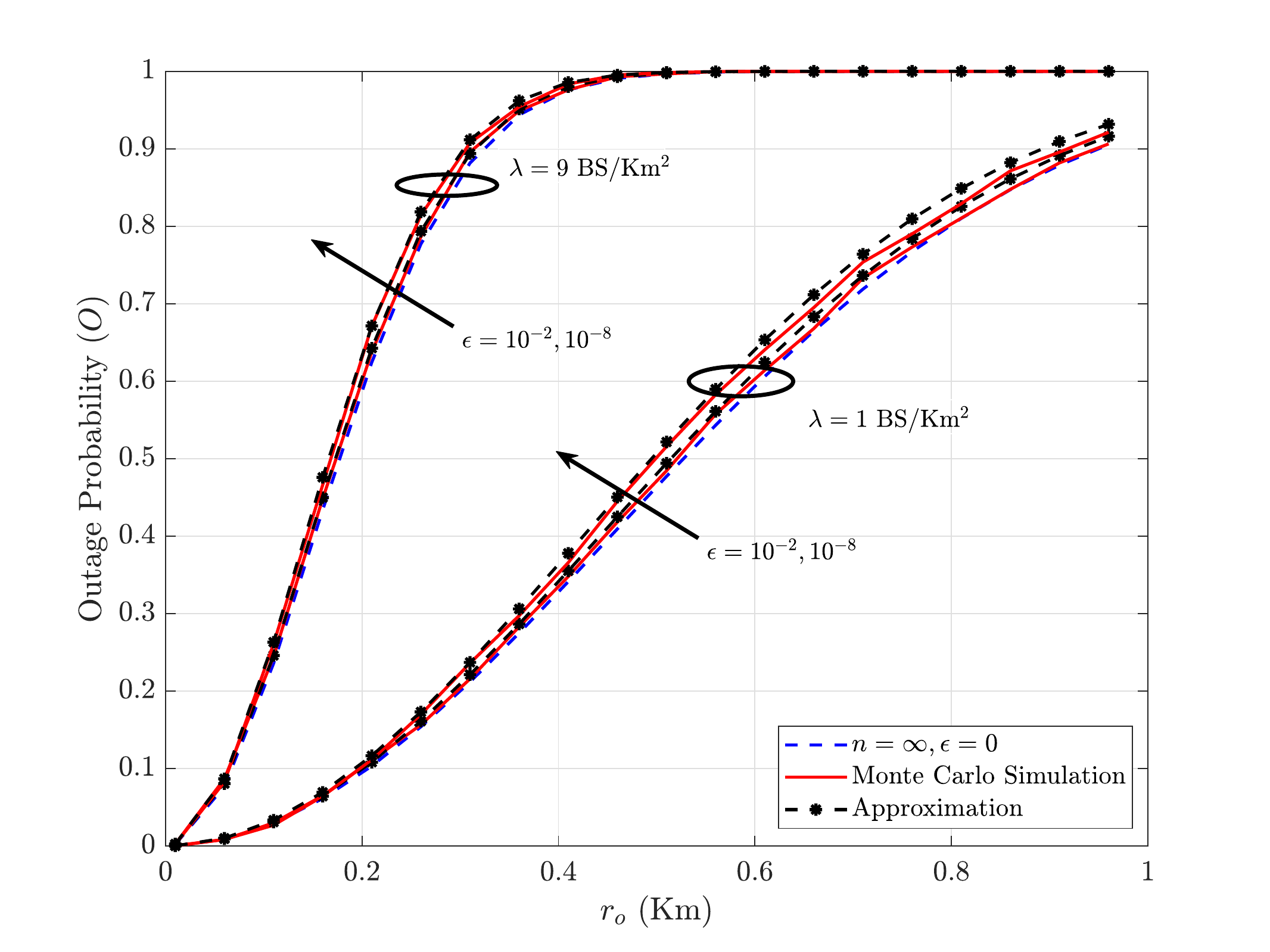}
         \caption{$n=2048$}
         \label{fig:Outage_App_2048}
     \end{subfigure}
     \caption{The outage probability of the large-scale network vs. $r_0$ at $\lambda=1,9\ \mathrm{BS/km^2}$, $\epsilon=10^{-2},10^{-8}$, and (a) $n=128$ (b) $n=2048$ }
     \label{fig:Outage_App}
\end{figure}


The effect of $\lambda$ and $r_0$ on average coding rate is shown in Fig. \ref{fig:r_o}, which shows the average coding rate at different BSs densities $\lambda=1,\ 3,\ 9\ \mathrm{BS/km^2}$ at $n=128$ and $\epsilon=10^{-5}$. The average coding rate decreases as $\lambda$ increases since increasing $\lambda$ increases interference. Moreover, at a specific value of $\lambda$, the average coding rate also decreases as $r_0$ increases since the strength of the desired signal decreases.

\begin{figure}
    \centering
    \includegraphics[width=\linewidth]{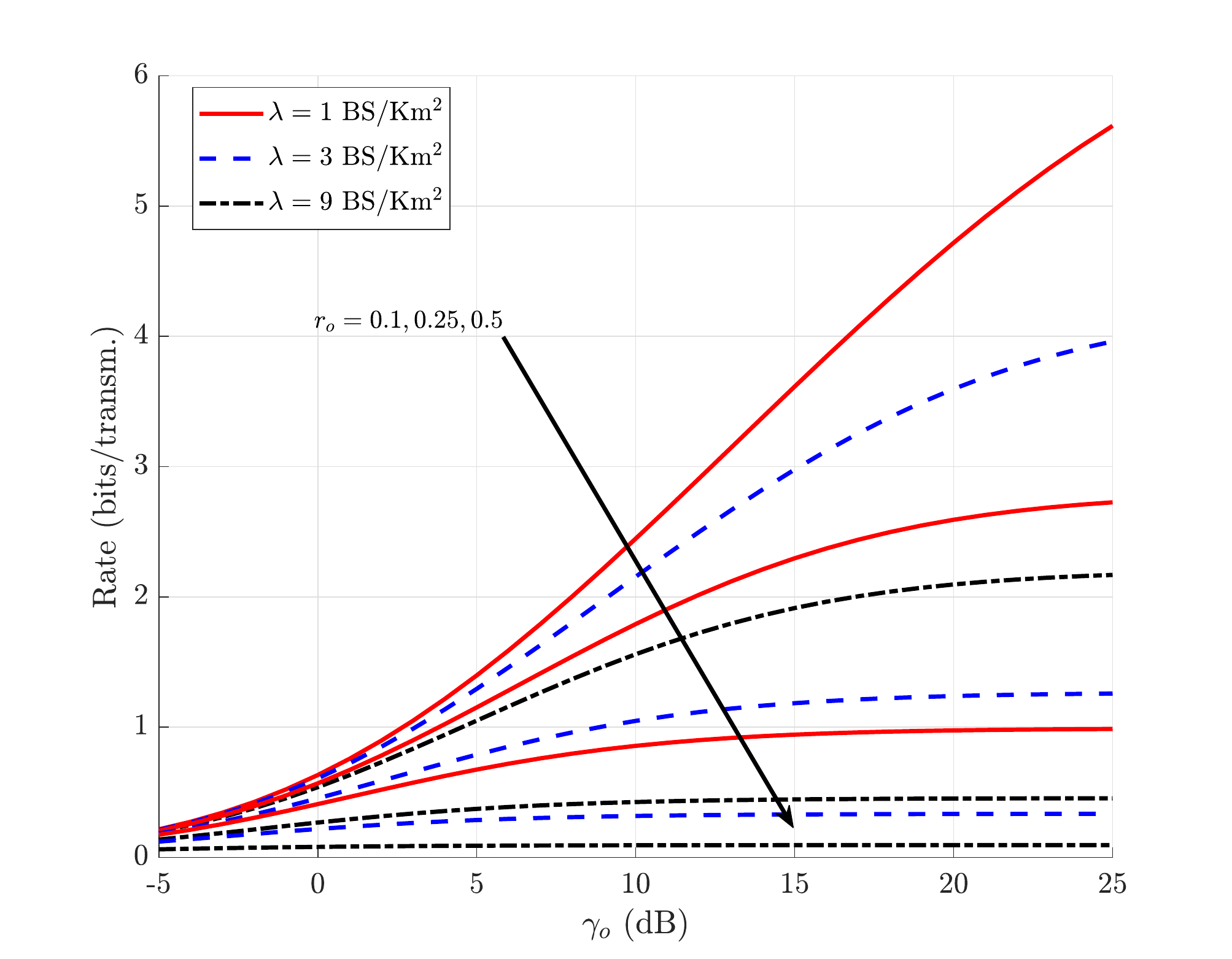}
    \caption{The average coding rate of the large-scale network vs. $\gamma_o$ for $r_0= 10,\ 250,\ 500$ and $\lambda=1,\ 3,\ 9\ \mathrm{BS/km^2}$ with $n=128$ and $\epsilon=10^{-5}$.}
    \label{fig:r_o}
\end{figure}

\section{Conclusion}
We have studied the average coding rate in the finite blocklength regime for a large-scale network using stochastic geometry, and provided a valid approximation for the outage probability of the system. Moreover, we evaluated the performance metrics as a function of the network parameters. The results show that existing results that use the capacity expression in the infinite blocklength regime for studying performance overestimate performance, especially when the blocklength is small or the desired frame error rate is small, which are important requirements of URLLC, IoT, and MTC. The provided expressions can be used to study large-scale networks in the presence of stringent delay or reliability constraints, such as IoT networks, MTC, etc. This work will further be extended to specific communication technologies like OMA, NOMA,..., etc.

\bibliographystyle{elsarticle-num}
\bibliography{Bibliography}
\end{document}